\documentclass[12pt,a4paper]{article}

\RequirePackage{etex}
\usepackage[utf8]{inputenc}
\usepackage{amsmath, amssymb, amsthm, amsfonts, mathrsfs}
\usepackage{bbm, commath}
\usepackage{multicol}
\usepackage{blindtext}
\usepackage{graphicx}
\usepackage{tikz-network}
\usetikzlibrary {positioning, patterns, calc}
\usetikzlibrary{graphs, graphs.standard, quotes}
\usepackage{authblk}
\usepackage{enumitem}
\usepackage[colorlinks=true,linkcolor=black,citecolor=black]{hyperref}
\usepackage{tabularx}
\usepackage{array}

\newtheorem{thm}{Theorem}

\newtheorem{prop}{Proposition}
\newtheorem{lem}{Lemma}
\newtheorem{cor}{Corollary}

\def \Zl {{\mathbbm Z}}

\def \Rl {{\mathbbm R}}
\def \Cl {{\mathbbm C}}
\def \e {{\bf e}}

\def \e {{\mathbf e}}

\newcommand{\ob}[1]{\left(#1\right)}
\newcommand{\cb}[1]{\left\lbrace #1\right\rbrace}
\newcommand{\tb}[1]{\left[#1\right]}
\newcommand{\mb}[1]{\left|#1\right|}

\title{\it Quantum Pair State Transfer on Isomorphic Branches}
\author[1]{Hiranmoy Pal}
\author[2]{Sarojini Mohapatra}
\affil[1,2]{National Institute of Technology Rourkela, India-769008 palh@nitrkl.ac.in}
\date{\today}

\begin{document}
	
	\maketitle

	
		\begin{abstract}
      The evolution of certain pair state in a quantum network with isomorphic branches, governed by the Heisenberg $XY$ Hamiltonian, depends solely on the local structure, and it remains unaffected even if the global structure is altered. All graphs which enable high-fidelity vertex state transfer can be considered as isomorphic branches of a quantum network to exhibit high-fidelity pair state transfer. The results are used to unveil the existence of pair state transfer in various graphs, including paths, cycles, and others.\\
      
	\end{abstract}

	
	\newpage
  \section{Introduction}
The continuous-time quantum walk \cite{farhi} on a quantum spin network, modeled by a graph $G$ with the Heisenberg $XY$ Hamiltonian as described in \cite{bac, bose, chr1}, is governed by the transition matrix $U_G(t)=\exp{\left(itA\right)},$ where $t\in\Rl$ and $A$ is the adjacency matrix of $G.$ The Laplacian matrix $L$ can be used instead of the adjacency matrix when defining the transition matrix $U_G(t)$ whenever $XYZ$ interaction model \cite{bose1} is adopted. However, for a regular graph, analyzing state transfer using either the adjacency or Laplacian model is equivalent, as both considerations yield the same information. The state associated with a vertex $a$ in $G$ is considered to be the characteristic vector $\e_a.$ Perfect state transfer (PST) \cite{bose, chr1} occurs at time $\tau$ between two distinct vertices $a$ and $b$ in $G$ whenever the fidelity of transfer $\mb{\e_a^TU_G(\tau)\e_b}^2$ attains its maximum value $1.$ Since PST is a rare phenomenon \cite{god2}, a relaxation known as pretty good state transfer (PGST) was introduced in \cite{god1, vin}. A graph $G$ is said to have PGST between the vertices $a$ and $b$ whenever the fidelity $\mb{\e_a^TU_G(t)\e_b}^2$ comes arbitrarily close to $1$ for appropriate choices of $t.$ The pair state associated with a pair of vertices $(a,b)$ is considered to be $\frac{1}{\sqrt{2}}\left(\e_a-\e_b\right)$. In the case of perfect pair state transfer (PPST), the fidelity of transfer between two linearly independent pair states assumes the maximum value $1.$ As like PGST, one may consider pretty good pair state transfer (pair-PGST) where pair states are considered instead of the vertex states. In the past two decades, several network topologies having high fidelity vertex state transfer has been observed, such as the Cartesian powers of the path on two or three vertices \cite{chr1}, the path $P_n$ on $n$ vertices \cite{chr1, god4, god1, bom}, circulant graphs \cite{ange1, bas1, pal4, pal6, pal7}, cubelike graphs \cite{ber, che}, Cayley graphs \cite{cao3, cao, pal3, tan}, distance regular graphs \cite{cou2}, Hadamard diagonalizable graphs \cite{joh1}, signed graphs \cite{brow}, corona products \cite{ack1}, blow-up graphs \cite{pal9}, etc. We find that these graphs can be considered as isomorphic branches to a larger network to enable various properties associated to pair state transfer. Using the Laplacian model, pair state transfer was first introduced by Chen et al. \cite{che1}, where it is shown that among paths and cycles, only the paths on three or four vertices and the cycle on four vertices admit Laplacian PPST, provided at least one of the pairs is an edge. Further investigation on pair state transfer can be found in cubelike graphs \cite{cao2}, Cayley graphs \cite{cao1, luo}, paths \cite{wang2}, vertex coronas \cite{wang}. Kim et al. \cite{kim} later generalized the notion of pair state transfer using $s$-pair states of the form $\e_a+s\e_b,$ where $s$ is a non-zero complex number. They analyzed the existence of perfect state transfer between $s$-states in complete graphs, cycles, and antipodal distance regular graphs having perfect vertex state transfer.

In this article, we study continuous-time quantum walks on graphs relative to the adjacency matrix. The subsequent sections are organized as follows. In Section \ref{2s2}, we present several results on the existence of PPST and pair-PGST in graphs, and briefly introduce equitable partition for weighted graphs. A non-trivial relationship between the transfer of vertex states and pair states is established in Section \ref{2s_3}, which provides a framework for constructing infinite families of graphs, such as trees, unicyclic graphs, and others, that exhibit PPST. Section \ref{2s3} includes the characterization of pair state transfer on paths and cycles. Although not all paths or cycles admit PPST, we show that addition of a few edges to it enable PPST in the resulting graph. In Section \ref{2s5}, we introduce perfect $(m,L)$-state transfer in graphs, and uncover a few related observations.

\section{Preliminaries}\label{2s2}
Let $G$ be an undirected and weighted graph having a vertex set $V(G)$ and a weight function $w: V(G)\times V(G)\to \mathbb{R}$ which is symmetric. If $G$ is a simple graph then $w(a,b)=1,$ whenever $a$ and $b$ are adjacent in $G$, otherwise $w(a,b)=0.$ The adjacency matrix $A$ of a graph $G$ is a square matrix whose rows and columns are indexed by the vertices of $G$ and $A_{a,b}=w(a,b).$ A graph $G$ is said to have perfect pair state transfer (PPST) between two linearly independent pair states $\frac{1}{\sqrt{2}}\ob{\e_a-\e_b}$ and $\frac{1}{\sqrt{2}}\ob{\e_c-\e_d}$ associated to $(a,b)$ and $(c,d),$ respectively, if there exists $\tau \in \mathbb{R}$ such that
\begin{equation}\label{equ1}
     U_G(\tau)\ob{\e_a-\e_b}=\gamma\ob{\e_c-\e_d}~~~\text{for some}~\gamma\in\mathbb{C}.
 \end{equation}
 In which case, we simply say that PPST occurs between pairs $(a,b)$ and $(c,d).$ If PPST occurs from $(a,b)$ to itself at time $\tau ~(\neq 0),$ then $G$ is said to be periodic at $(a,b).$ The spectrum $\sigma(G)$ is the set of all eigenvalues of the adjacency matrix of $G.$ Suppose $E_{\lambda}$ is the orthogonal projection onto the eigenspace corresponding to eigenvalue $\lambda.$ If $\lambda_1, \lambda_2,\ldots, \lambda_d$ are the distinct eigenvalues of $A,$ then the spectral decomposition of $A$ is given by
 \[A=\sum\limits_{j=1}^d\lambda_j E_{\lambda_j}.\] The idempotents satisfy $E_{\lambda_j}E_{\lambda_k}=0,$ whenever $j\neq k,$ and $E_{\lambda_1}+E_{\lambda_2}+\cdots+E_{\lambda_d}=I,$ where $I$ is the identity matrix. The eigenvalue support of the pair state  $\frac{1}{\sqrt {2}}\left(\e_a-\e_b\right)$ associated to $(a,b)$ is defined by $\sigma_{ab}=\{\lambda\in \sigma(G):E_{\lambda}(\e_a-\e_b)\neq0\}.$ If PPST occurs between $(a,b)$ and $(c,d),$ then \eqref{equ1} implies that the support of $(a,b)$ and $(c,d)$ are identical. The states associated to $(a,b)$ and $(c,d)$ are called strongly cospectral if and only if 
\[E_{\lambda}(\e_a-\e_b)=\pm E_{\lambda}(\e_c-\e_d)\]
holds for all $\lambda\in \sigma(G)$. It is well known that if PPST occurs between $(a,b)$ and $(c,d)$ in $G$ then the associated states are strongly cospectral \cite[Theorem 2.3]{kim}. The strong cospectrality condition together with the spectral decomposition of $A$ yields 
\begin{eqnarray*}
    \ob{\e_a-\e_b}^T A^k\ob{\e_a-\e_b}= \ob{\e_c-\e_d}^TA^k\ob{\e_c-\e_d},
\end{eqnarray*}
for all non-negative integers $k$. If $G$ is simple then $(a,b)$-th entry of $A^k$ represents the number of walks of length $k$ from $a$ to $b$. Let $N(a)$ denote the set of all neighbours of vertex $a.$ Considering $k=2,$  we arrive at the following conclusion.
    \begin{prop}\label{2c1}
     If $(a,b)$ and $(c,d)$ are strongly cospectral in a simple graph $G,$ then 
     \[\mb{N(a)\cup N(b)}-\mb{N(a)\cap N(b)}=\mb{N(c)\cup N(d)}-\mb{N(c)\cap N(d)}.\]
    \end{prop}
 It is worth mentioning that periodicity is necessary for a pair state to exhibit PPST \cite[Theorem 2.5]{kim}. 
 The following result provides necessary and sufficient conditions for a pair state to be periodic.
\begin{thm}\cite{kim}\label{2t_1}
   Let $G$ be a graph with adjacency matrix $A.$ Suppose $S$ is the eigenvalue support of a pair of vertices $(a,b)$ in $G.$ Then $(a,b)$ is periodic if and only if either:
\begin{enumerate}
\item All eigenvalues in $S$ are integers, or
\item There exists a square-free integer $\Delta>1$ and an integer $c$ such that each eigenvalue
 in $S$ is in the form $\frac{c+d\sqrt{\Delta}}{2}$ for some integer $d.$
\end{enumerate} 
\end{thm} 

Since there are only a few paths and cycles exhibiting PPST (see Section \ref{2s3}), we consider a generalization to it as follows. A graph $G$ is said to exhibit pair-PGST between two linearly independent pair sates $\frac{1}{\sqrt{2}}\ob{\e_a-\e_b}$ and $\frac{1}{\sqrt{2}}\ob{\e_c-\e_d}$ associated to $(a,b)$ and $(c,d),$ respectively, if there exists a sequence $t_k\in\Rl$ such that
 \begin{eqnarray}\label{ppgst}
 \lim_{k\to\infty}U_G\ob{t_k}\ob{\e_a-\e_b}=\gamma \ob{\e_c-\e_d}~~~\text{for some}~\gamma\in\mathbb{C}.
 \end{eqnarray}
In \cite[Lemma 3.3]{bom2}, van Bommel showed that if a graph has pretty good state transfer between a pair of arbitrary states, then they are strongly cospectral. In particular, if a graph $G$ has pretty good pair state transfer between $(a,b)$ and $(c,d)$, then the associated pair states are strongly cospectral.
 
\subsection{Algebraic Properties}
 An automorphism $f$ of a weighted graph $G$ is a bijection on the vertex set $V(G)$ satisfying $w(a,b)=w\ob{f(a),f(b)}$ for all $a,b\in V(G).$ The set of all automorphisms of $G$ is denoted by $\text{Aut}(G).$ If $P$ is the matrix of the automorphism $f,$ then $P$ commutes with the adjacency matrix $A.$ Since the transition matrix $U_G(t)$ is a polynomoal in $A,$ the matrix $P$ commutes with $U_G(t)$ as well. Let $G$ have pair-PGST between $(a,b)$ and $(c,d),$ then \eqref{ppgst} gives
 \begin{eqnarray}\label{aut}
 \displaystyle\lim_{k\to\infty} U_G(t_k)P\ob{\e_a-\e_b}=\gamma P\ob{\e_c-\e_d}.
 \end{eqnarray}
Since the sequence $U_G(t_k)\ob{\e_a-\e_b}$ can not have two different limits, we conclude that if $P$ fixes
 $\ob{\e_a-\e_b}$ then $P$ must fix $\ob{\e_c-\e_d}$ as well. The following result is analogous to the observations in \cite[Lemma 4.3]{che1} and \cite[Lemma 4.4]{che1}.
 \begin{lem}\label{al1}
 Let a graph $G$ admits pretty good pair state transfer between $(a,b)$ and $(c,d).$ Then the stabilizer of $(a,b)$ is the same as the stabilizer of $(c,d)$ in $\text{Aut}(G).$ Moreover, all pairs in the orbit of $(a,b)$ under $\text{Aut}(G)$ have pretty good pair state transfer.
 \end{lem}
Another observation follows from the fact that if $\e_a$ is the only characteristic vector fixed by $P,$ then subtracting \eqref{ppgst} from \eqref{aut} yields
 \begin{eqnarray}\label{equ4}
 \displaystyle\lim_{k\to\infty} U_G(t_k)\ob{\e_b-P\e_b}=\gamma \tb{P\ob{\e_c-\e_d}-\ob{\e_c-\e_d}},
 \end{eqnarray}
 which is absurd whenever $a$ is distinct from $c$ and $d.$ 
 \begin{lem}\label{al3}
  If an automorphism of a graph fixes only one vertex $a,$ then there is no pair-PGST between $(a,b)$ and $(c,d)$ whenever $a$ is distinct from $c$ and $d.$
  \end{lem} 
 
Consider a Cayley graph defined over a finite abelian group $\Gamma$ of odd order with a connection set $S\subseteq\Gamma$ satisfying $\left\lbrace -s:s\in S\right\rbrace=S$. The Cayley graph, denoted by $\text{Cay}(\Gamma, S),$ has the vertex set $\Gamma$ where two vertices $a$ and $b$ are adjacent if and only if $a-b\in S$. The map sending $\zeta$ in $\Gamma$ to its inverse $\zeta^{-1}$ is an automorphism of $\text{Cay}(\Gamma, S)$ fixing only the identity $0.$ Hence, there is no pair-PGST from $(0,b)$ to $(c,d)$ whenever $0\notin\cb{c,d}$ by Lemma \ref{al3}. Suppose there is pair-PGST between $(0,a)$ and $(0,b)$ in $\text{Cay}(\Gamma, S).$ Since there is an automorphism of $\text{Cay}(\Gamma, S)$ fixing only the vertex $a,$ we arrive at a similar contradiction as in \eqref{equ4}. Finally, since a Cayley graph is vertex-transitive, there is no pair-PGST in $\text{Cay}(\Gamma, S),$ whenever $\Gamma$ is an abelian group of odd order.
 
  \begin{lem}\label{al2}
  There is no pair-PGST in a Cayley graph over an abelian group of odd order.
  \end{lem}

\subsection{Equitable Partition}
A partition $\Pi$ of the vertex set of a wighted graph $G$ with disjoint cells $V_1, V_2, \ldots, V_d$ is said to be equitable if
\[\sum_{b \in V_k} w(a,b)=c_{jk}~~~\text{for each}~ a\in V_j,\]
where $c_{jk}\in \Rl$ is a constant depending only on the cells $V_j$ and $V_k$. The characteristic matrix $\mathcal{C}$ associated to $\Pi$ is an $n \times d$ matrix where the columns represent the normalized characteristic vectors of $V_1, V_2,\dots, V_d.$ The graph with the $d$ cells of $\Pi$ as its vertices having adjacency matrix $A_\Pi,$ where $\ob{A_\Pi}_{j,k}=\sqrt{c_{jk}c_{kj}},$ is called symmetrized quotient graph of $G$ over $\Pi,$ and is denoted by $G/\Pi.$ The following result plays a crucial role in establishing the main result in Section \ref{2s_3}.

\begin{prop}\cite{bac, god0}\label{ep1}
 Let $G$ be a weighted graph with adjacency matrix $A.$ If $\Pi$ is an equitable partition of $G$ having characteristic matrix $\mathcal{C}$ then $A\mathcal{C}=\mathcal{C}A_{\Pi},$ where $A_{\Pi}$ is the adjacency matrix of the symmetrized quotient graph $G/\Pi.$ Moreover, $v$ is an eigenvector of $A_{\Pi}$ if and only if $\mathcal{C}v$ is an eigenvector of $A.$ 
\end{prop}

The eigenvectors of $A$ can be regarded as a real valued function on the vertex set of $G.$ An eigenvector $\mathcal{C}v$ arising from the equitable partition $\Pi$ is a linear combination of the columns of $\mathcal{C}.$ Since the orbits of a group of automorphisms of a graph $G$ form an equitable partition \cite{god0}, the eigenvectors of $A$ arising from the equitable partition are constant on the orbits.

\section{Pair state transfer on isomorphic branches}\label{2s_3}

In this section, we present a result that establishes a non-trivial relationship between vertex state and pair state transfer. Suppose a weighted graph $X_1$ is isomorphic to $X_2$ where $f: X_1\to X_2$ is an isomorphism satisfying $w(a,b)=w\ob{f(a),f(b)},$ for all pair of vertices $a$ and $b$ in $X_1.$ Consider a connected graph $G$ where $X_1\cup X_2,$ the disjoint union of $X_1$ and $X_2,$ appears as an induced subgraph in such a way that $w(a,y)=w\ob{f(a),y},$ for all $a\in V\ob{X_1}$ and $y\in V(G)\setminus V\ob{X_1\cup X_2}.$ In which case, we say that $X_1$ and $X_2$ are isomorphic branches of the graph $G.$ The following result describes the evolution of certain pair states in a graph with isomorphic branches.
\begin{lem}\label{2t1}
 Let $X_1$ and $X_2$ be isomorphic branches of a graph $G$ with an isomorphism $f: X_1\to X_2.$ If $U_G(t)$ is the transition matrix of $G$ then for each vertex $a$ in $X_1$ 
{\small
\begin{eqnarray}\label{e00}
U_G(t)\ob{\e_a-\e_{\hat{f}(a)}}=(I-P)\sum\limits_{\lambda\in \sigma\left(X_1\right)} \exp{(it\lambda)}E_{\lambda}\e_a,
\end{eqnarray}}
where $P$ is the matrix of the automorphism $\hat{f}$ of $G$ that switches each vertex of $X_1$ to its $f$-isomorphic copy in $X_2,$ and fixing all other vertices of $G.$ Assuming $F_{\lambda}$ to be the eigenprojector of $X_1$ corresponding to $\lambda$ in $\sigma\left(X_1\right),$ the matrix $E_{\lambda}$ is given by
\[E_{\lambda}=\frac{1}{2}\begin{bmatrix}
    ~~F_{\lambda} & -F_{\lambda} & \mathbf{0}\\
    -F_{\lambda} & ~~F_{\lambda} & \mathbf{0}\\
    ~~\mathbf{0} & ~~\mathbf{0} & \mathbf{0}
\end{bmatrix}. \]
\end{lem}
\begin{proof}
The spectral decomposition of the transition matrix of $G$ gives 
{\small
\begin{eqnarray}\label{e0}
    U_G(t)\left(\e_a-\e_{\hat{f}(a)}\right)=\sum\limits_{\lambda\in \sigma\left(G\right)} \exp{(it\lambda)}E_{\lambda}\left(\e_a-\e_{\hat{f}(a)}\right),
\end{eqnarray}}
where $E_{\lambda}$ is the eigenprojector associated to $\lambda\in \sigma\left(G\right).$ Consider the equitable partition of $G$ formed by the orbits of the automorphism $\hat{f}.$ The entries of all eigenvectors of $G$ arising from the equitable partition are constant on the orbits of $\hat{f}.$ Consequently, the eigenvectors of $G$ that contribute to the sum on the right of \eqref{e0} is of the form $\left(x^T,-x^T,0 \right)^T,$ where $x$ is an eigenvector of $X_1.$ The adjacency matrix of $G$ commutes with $P,$ and hence $PE_{\lambda}=E_{\lambda}P.$ Since $\e_a-\e_{\hat{f}(a)}=(I-P)\e_a,$ we have the desired result. 
\end{proof}

In Lemma \ref{2t1}, the orbits of $\hat{f}$ forms an equitable partition of $G,$ where $\cb{a, \hat{f}(a)}$ forms a cell for every $a\in X_1$ and the remaining vertices are in singleton cells. Graphs with high fidelity pair sate transfer are obtained using the following result.

\begin{thm}\label{2c2}
    Suppose the premise of Lemma \ref{2t1} holds, and $\mathcal{B}=\cb{\frac{1}{\sqrt{2}}\ob{\e_a-\e_{\hat{f}(a)}}~\bigm\vert~\text{where } a\in X_1}.$ Consider an equitable partition $\Pi$ of $G$ where $\cb{a, \hat{f}(a)}$ forms a cell for every $a\in X_1.$ Let $Q$ be the matrix whose columns are the vectors in $\mathcal{B}$ followed by the normalized characteristic vectors of $\Pi.$ Then for all $t\in\Rl$
     \[Q^T U_G(t)Q=\begin{bmatrix}
     U_{X_1}(t) & \mathbf{0}\\
     \mathbf{0} & U_{G/\Pi}(t)
     \end{bmatrix}.\]
    \end{thm}

\begin{proof}
Let $\e_a$ and $\check{\e}_a$ be the characteristic vectors of a vertex $a$ in $G$ and $X_1$, respectively. Since the permutation matrix $P$ commutes with $U_G(t),$ for each vertex $b$ in $X_1$ 
\[\frac{1}{2}\ob{\e_b-\e_{\hat{f}(b)}}^T U_G(t) \ob{\e_a-\e_{\hat{f}(a)}}= \frac{1}{2}\e_b^T(I-P) U_G(t) (I-P)\e_a = \e_b^T U_G(t) \ob{\e_a-\e_{\hat{f}(a)}}.\]
Applying Lemma \ref{2t1} gives  
  \begin{align*}
     \e_b^T U_G(t) (\e_a-\e_{\hat{f}(a)})&= \sum_{\lambda\in\sigma(X_1)}\exp{(it\lambda)}~\e_b^TE_\lambda(I-P)\e_a \\ &= \frac{1}{2}\sum_{\lambda\in\sigma(X_1)}\exp{(it\lambda)}~\check{\e}_b^T\begin{bmatrix}
    F_{\lambda} & -F_{\lambda} & \mathbf{0}
\end{bmatrix} \begin{bmatrix}
    ~~\check{\e}_a\\ -\check{\e}_a\\ ~~0
\end{bmatrix}\\
&= \check{\e}_b^T U_{X_1}(t) \check{\e}_a.
  \end{align*}
  It follows that, if $W=\text{span}\ob{\mathcal{B}}$ then the transition matrix $U_G(t)$ is $W$-invariant. The matrix of the restriction operator $U_G(t)\bigl\vert_W$ relative to $\mathcal{B}$ becomes $U_{X_1}(t).$ Suppose $\mathcal{C}$ is the matrix whose columns are the normalized characteristic vectors of $\Pi.$ The columns of $\mathcal{C}$  forms a basis of $W^{\perp},$ the subspace orthogonal to $W.$ If the adjacency matrices of $G$ and $G/\Pi$ are $A$ and $A_{\Pi}$, respectively, then $A\mathcal{C}=\mathcal{C}A_{\Pi}.$ Consequently, we obtain $U_G(t)\mathcal{C}=\mathcal{C}U_{G/\Pi}(t),$ as observed in \cite{bac, ge}. This completes the proof.
\end{proof}

\begin{figure}
\begin{center}
\begin{tikzpicture}[scale=.5,auto=left]
                       \tikzstyle{every node}=[circle, thick, fill=white, scale=0.6]
                       
		        \node[draw] (1) at (1.5,0) {$u$};		        
		        \node[draw,minimum size=0.65cm, inner sep=0 pt] (2) at (-2.5, 1.5) {$a$};
		        \node[draw,minimum size=0.65cm, inner sep=0 pt] (3) at (-1.5, 2.5) {$b$};
		        \node[draw,minimum size=0.65cm, inner sep=0 pt] (4) at (-2.5, -1.5) {$c$};
		        \node[draw,minimum size=0.65cm, inner sep=0 pt] (5) at (-1.5, -2.5) {$d$};		        
	
				\node at (-4.25, 2) {$X_1$};
				\node at (-4.25,-2) {$X_2$};
				\node at (6.3,0.5) {$H$};
				
				\draw[dashed] (-2,2) circle (1.5 cm);
				\draw[dashed] (-2,-2) circle (1.5 cm);
				\draw[dashed] (3.2,0.5) circle (2.5 cm);
								
				\draw [thick, black!70] (1)--(3)--(2);
				\draw [thick, black!70] (1)--(5)--(4);

				\draw[thick, black!70] (1)..controls (3,4) and (8,0)..(1);

		        \node[draw,minimum size=0.65cm, inner sep=0 pt] (6) at (-7,0) {$u$};		        
		        \node[draw,minimum size=0.65cm, inner sep=0 pt] (7) at (-11,1.5) {$a$};
		        \node[draw,minimum size=0.65cm, inner sep=0 pt] (8) at (-10,2.5) {$b$};
		        \node[draw,minimum size=0.65cm, inner sep=0 pt] (9) at (-11,-1.5) {$c$};
		        \node[draw,minimum size=0.65cm, inner sep=0 pt] (10) at (-10,-2.5) {$d$};			

 	                         \draw[dashed] (-10.5,2) circle (1.5 cm);
				\draw[dashed] (-10.5,-2) circle (1.5 cm);
				
				\node at (-12.75, 2) {$X_1$};
				\node at (-12.75,-2) {$X_2$};
				
                \draw [thick, black!70] (7)--(8)--(6)--(10)--(9);				
				\end{tikzpicture}
				
\end{center}	
\caption{\label{fig1} The path $P_5$ and its perturbation.}
\end{figure}
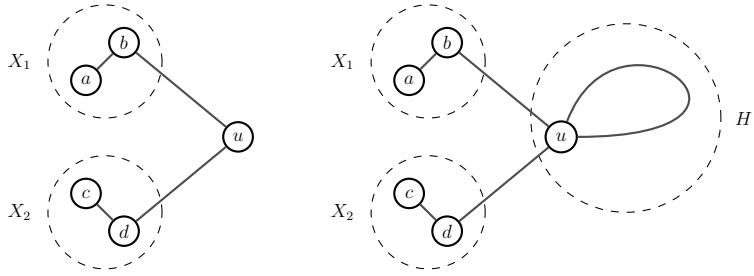

In Theorem \ref{2c2}, we observe that the evolution of certain pair states in a graph $G$ having isomorphic branches depends only on the local structure of $G.$ It is evident that all properties related to vertex state transfer in $X_1$ have natural consequences in terms of the pair states in $G.$ Consider the path $P_5$ as given in Figure \ref{fig1}, where $X_1$ and $X_2$ appear as isomorphic branches. It is well known that the path $X_1$ on two vertices admits PST at time $\frac{\pi}{2}$ between the end vertices $a$ and $b.$ Consequently, PPST occurs in $P_5$ between $(a,c)$ and $(b,d)$ at the same time. In fact, if the middle vertex of $P_5$ is identified with any vertex of a graph $H$ then PPST between $(a,c)$ and $(b,d)$ remains unaffected. In particular, considering $H$ a tree or an unicyclic graph, we have the following.

\begin{cor}
 There are infinitely many trees and unicyclic graphs exhibiting perfect pair state transfer. \end{cor}
 
 It is worth mentioning that among all trees only the path on two or three vertices exhibits PST \cite{cou7}. In \cite{cou9}, the author finds that a graph with $m$ edges never have PST from vertex $a$ whenever the eccentricity $\epsilon_a$ of $a$ satisfies $(\epsilon_a)^3\geq 54m.$ However, in the case of PPST, there is no such bound on eccentricity of the vertices or the graph diameter. One may wonder whether there is a simple graph $G$ exhibiting PST between a pair of vertices $a$ and $b,$ and PGST between another pair $c$ and $d$ where there is no PST. The answer is positive when considering pair states instead of the vertex states. Consider the graph $G$ given in Figure \ref{2Fig_2}. One may apply Theorem \ref{2c2} and observe that PPST occurs in $G$ between $(1,3)$ and $(5,6)$. Here $G$ exhibits pair-PGST between $(2,4)$ and $(9,12)$ as well. Since there is no PST in $P_4,$ we have no PPST between $(2,4)$ and $(9,12).$ 

 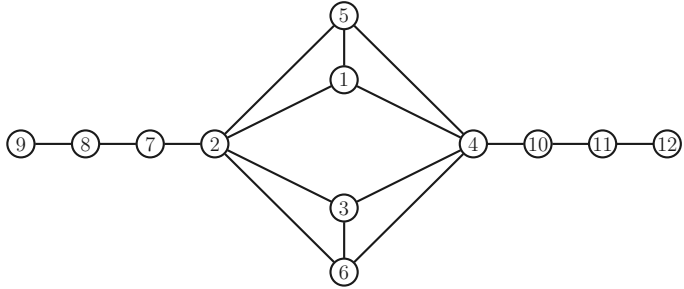
\begin{figure}
		\centering
                    \begin{tikzpicture}[scale=1.7,auto=left]
                       \tikzstyle{every node}=[circle, thick, black!90, fill=white, scale=0.65]

                \node[draw,minimum size=0.55cm, 
                inner sep=0 pt] (9) at (0, 0) {$9$};
				 \node[draw,minimum size=0.55cm, inner sep=0 pt] (8) at (0.5,0) {$8$};
				 \node[draw,minimum size=0.55cm, 
                 inner sep=0 pt] (7) at (1, 0) {$7$};
                 
                 \node[draw,minimum size=0.55cm, inner sep=0 pt] (2) at (1.5, 0) {$2$};
            
                   \node[draw,minimum size=0.55cm, 
                inner sep=0 pt] (4) at (3.5,0) {$4$};
				 \node[draw,minimum size=0.55cm, inner sep=0 pt] (10) at (4,0) {$10$};
				 \node[draw,minimum size=0.55cm, inner sep=0 pt] (11) at (4.5, 0) {$11$};
     \node[draw,minimum size=0.55cm, inner sep=0 pt] (12) at (5, 0) {$12$};
     
                 \node[draw,minimum size=0.55cm, inner sep=0 pt] (3) at (2.5, -0.5) {$3$};

                 \node[draw,minimum size=0.55cm, 
                inner sep=0 pt] (6) at (2.5, -1) {$6$};
				 \node[draw,minimum size=0.55cm, inner sep=0 pt] (1) at (2.5, 0.5) {$1$};
     \node[draw,minimum size=0.55cm, inner sep=0 pt] (5) at (2.5, 1) {$5$};

      \draw[thick, black!90] (1)-- (2)--(3)--(4)--(1)--(5);
  \draw[thick, black!90]  (3)--(6)--(4)--(10)--(11)--(12);
   \draw[thick, black!90] (6)-- (2)--(7)--(8)--(9);
  \draw[thick, black!90]  (2)--(5)--(4);
\end{tikzpicture}	
		\caption{A graph having PPST and pair-PGST.}
  \label{2Fig_2}
	\end{figure}

Suppose $K_n$ is a complete graph on $n$ vertices. All permutations of the $n$ vertices of $K_n$ forms an automorphism. One can deduce using Lemma \ref{al3} that there is no pair-PGST in $K_n.$ Let $a,$ $b,$ $c,$ $d$ be distinct vertices in $K_n,$ where $n\geq 5,$ which form a cycle $C_4$ with $a$ and $c$ as antipodal vertices. Observe that the edges $(a,c)$ and $(b,d)$ are appearing as isomorphic branches of $K_n\backslash C_4.$ Theorem \ref{2c2} applies here to find that PPST occurs between $(a,b)$ and $(c,d).$ In fact, we have the following conclusion.
\begin{cor}
   Let $K_n$ be a complete graph on $n$ vertices. The removal of any number of disjoint cycles on four vertices from $K_n,$ where $n\geq 5,$ gives PPST from every pair of non-adjacent vertices in the resulting graph. 
\end{cor}

It is well known that $C_4,$ the cycle of length $4,$ admits PST at $\frac{\pi}{2}$ between antipodal vertices. We find a family of Cayley graphs exhibiting PPST, where multiple copies of $C_4$ appear as isomorphic branches. Let $\Zl_m$ denote the group of integers modulo $m.$ Consider $S=\cb{(0,\pm 1),(\pm 1,0),(\pm 1,2),(\pm m,0)}\subset \mathbb{Z}_{4m}\times \mathbb{Z}_4$ with $m>1.$ The Cayley graph $\text{Cay}\ob{\mathbb{Z}_{4m}\times \mathbb{Z}_4, S}$ has the vertex set $\mathbb{Z}_{4m}\times \mathbb{Z}_4,$ where two vertices $a$ and $b$ are adjacent whenever $a-b\in S.$ The edges of $\text{Cay}\ob{\mathbb{Z}_{4m}\times \mathbb{Z}_4, S}$ associated to $(\pm m,0)\in S$ forms $4m$ disjoint cycles of length $4,$ where the vertices $(j,k)$ and $(2m+j,k),$ for all $(j,k)\in\mathbb{Z}_{4m}\times \mathbb{Z}_4,$ appear as antipodal vertices in each copy of $C_4.$ The vertices $(j,k)$ and $(j,k+2)$ are appearing in two distinct copies of $C_4$, say $X_1$ and $X_2.$ Both $(j,k)$ and $(j,k+2)$ are adjacent to $(j, k\pm 1),~(j\pm 1, k),~(j\pm 1, k+2)$ in $V(G)\setminus V\ob{X_1\cup X_2}.$ Now, Theorem \ref{2c2} applies to find a family of Cayley graphs having PPST (or Laplacian PPST) from a pair of non-adjacent vertices.
\begin{cor}
    Let $m>1$ be an integer, and let $S=\cb{(0,\pm 1),(\pm 1,0),(\pm 1,2),(\pm m,0)}$ be a subset of $ \mathbb{Z}_{4m}\times \mathbb{Z}_4.$ The graph $\text{Cay}\ob{\mathbb{Z}_{4m}\times \mathbb{Z}_4, S}$ exhibits perfect pair state transfer between $((j,k),(j,k+2))$ and $((2m+j,k),(2m+j,k+2)),$ for all $(j,k)\in\mathbb{Z}_{4m}\times \mathbb{Z}_4.$
\end{cor}

A vertex $a$ in a graph $G$ is said to be $C$-sedentary if for some constant $0 < C \leq 1,$
\[\inf_{t>0}~\mb{\e_a^TU_G(t)\e_a}\geq C.\]
The notion of sedentary family of graphs was first introduced by Godsil \cite{god5}. It is observed in \cite[Proposition 2]{mon} that a $C$-sedentary vertex in a graph does not have PGST. Theorem \ref{2c2} applies to identify graphs having $C$-sedentary pair states where there is no pair-PGST.

\begin{cor}\label{3c3}
    Suppose the premise of Lemma \ref{2t1} holds. Let $X_1$ has a vertex $a$ that is $C$-sedentary. Then there is no pretty good pair state transfer from $(a,f(a))$ in $G.$
\end{cor}
In a complete graph $K_n$ on $n$ vertices where $n\geq 3,$ one observes $\mb{\e_a^TU_{K_n}(t)\e_a}\geq 1-\frac{2}{n},$ for all vertices $a.$ Hence, each vertex in $K_n$ is $\ob{1-\frac{2}{n}}$-sedentary. Now, consider the graph $G=K_n+K_1+K_n,$ for $n\geq 3,$ where all vertices of two disjoint copies of $K_n$ are joined with an isolated vertex by edges. Since the pair state associated to $(a,b)$ becomes an eigenvector of $G$ whenever $a$ and $b$ lies in the same copy of $K_n,$ there is no PPST from $(a,b)$ (see \cite{kim}). In the remaining cases, where both $a$ and $b$ have degree $n,$ Corollary \ref{3c3} applies to conclude that $G$ does not exhibit PPST. Finally, if the vertex $x$ of degree $2n$ is paired with another vertex $b$ in $G,$ then the eigenvalue support of $(x,b)$ contains the eigenvalues $n-1$ and $\frac{(n-1)\pm\sqrt{(n+1)^2+4n}}{2}.$ Since periodicity is necessary for the existence of PPST, we find using Theorem \ref{2t_1} that there is no PPST in $G.$

A graph $G$ is said to have fractional revival from a vertex $a$ whenever the transition matrix of $G$ maps the state associated with $a$ to a linear combination of the states of a collection of vertices containing the initial one \cite{bern1, cha2, gan1, mon1}. One may consider fractional revival on $G$ from a pair state as well. A graph $G$ admits fractional revival at time $\tau$ between pair states $\frac{1}{\sqrt{2}}\ob{\e_a-\e_b}$ and $\frac{1}{\sqrt{2}}\ob{\e_c-\e_d}$ associated to $(a,b)$ and $(c,d),$ respectively, if for some $\alpha, \beta \in\mathbb{C}$ with $\beta\neq 0$, we have
\[U_G(\tau)\ob{\e_a-\e_b}=\alpha\ob{\e_a-\e_b}+\beta \ob{\e_c-\e_d}.\]
It is well known that $P_4$ exhibits $\ob{-\cos{\frac{\pi}{\sqrt{5}}}, -i\sin{\frac{\pi}{\sqrt{5}}}}$-revival at $\frac{2\pi}{\sqrt{5}}$ between the end vertices. Then Theorem \ref{2c2} applies to conclude that $P_9$ exhibits fractional revival between $(1,9)$ and $(4,6)$ at the same time. Using Theorem \ref{2c2}, one may construct families of graphs having fractional revival from a pair state to another as below.

\begin{cor}
Suppose the premise of Theorem \ref{2c2} holds. If $X_1$ admits fractional revival between vertices $a$ and $b,$ then there is fractional revival between $\frac{1}{\sqrt{2}}\ob{\e_a-\e_{f(a)}}$ and $\frac{1}{\sqrt{2}}\ob{\e_b-\e_{f(b)}}$ in $G.$ 
 \end{cor}

\section{Pair state transfer on special classes of graphs}\label{2s3}
The results uncovered in Section \ref{2s_3} play a significant role in characterizing pair state transfer on special classes of graphs, including paths, cycles, and others. First we present a complete characterization of PPST in paths. Then, we unfold a sufficient condition for pair-PGST in paths, which leads to a complete characterization of pair-PGST in cycles. These insights guide us to discover new families of graphs exhibiting pair state transfer.

\subsection{Paths}

After the pioneering work of Bose \cite{bose}, which demonstrates that PST occurs between the end vertices of $P_2$, significant progress has been made in the characterization of PST and PGST on paths. It is observed in \cite{chr1, god1} that $P_n$ on $n$ vertices exhibits PST if and only if $n=2, 3.$ Later, Chen et al. \cite{che1} showed that Laplacian PPST occurs on $P_n$ from a pair of adjacent vertices if and only if $n=3, 4.$ We examine the existence of PPST on paths with respect to the adjacency matrix including the case where both pairs are not edges. The eigenvalues \cite{bro} of $P_n$ are $\lambda_j=2\cos{\left(\frac{j\pi}{n+1}\right)},$ for $1\leq j \leq n,$ and the corresponding eigenvectors are $\ob{\beta_1,\beta_2,\ldots,\beta_n}^T,$ where $\beta_k=\sin\left(\frac{jk\pi}{n+1}\right). $ If $\lambda_2\notin\sigma_{ab},$ the support of a pair $(a,b),$ then $E_{\lambda_2}(\e_a-\e_b)=0.$ Equivalently,  
\[\sin{\ob{\frac{2a\pi}{n+1}}}-\sin{\ob{\frac{2b\pi}{n+1}}}=0,\]
which implies that $\cos{\ob{\frac{(a+b)\pi}{n+1}}}\sin{\ob{\frac{(a-b)\pi}{n+1}}}=0.$  Since $\frac{a-b}{n+1}$ is not an integer, $\frac{2(a+b)}{n+1}$ must be an odd integer. The maximum value of $a+b$ is $2n-1$, and hence the only possibility is the case when $a+b\in\cb{\frac{n+1}{2},~\frac{3(n+1)}{2}}$ and $n$ is odd. Therefore, we have the following observation.
\begin{lem}\label{2l1}
    Let $a$ and $b$ be two vertices in $P_n$ and $\lambda_2=2\cos\frac{2\pi}{n+1}.$ Then
    \begin{enumerate}
        \item $\lambda_2\in\sigma_{ab}$ for all pairs $(a,b),$ whenever $n$ is even. 
        \item $\lambda_2\in\sigma_{ab}$  whenever $n$ is odd and $a+b\notin \cb{\frac{n+1}{2},~\frac{3(n+1)}{2}}.$
    \end{enumerate}
\end{lem}

If there is an automorphism $\eta$ of $G$ such that $\eta(a)=c$ and $\eta(b)=d,$ then both $(a,b)$ and $(c,d)$ have the same eigenvalue support. As a consequence, if $a, b, c, d$ are vertices in the path $P_n$ on odd number of vertices such that $a+b=\frac{n+1}{2}$ and $c+d=\frac{3(n+1)}{2},$ then the support of $(a,b)$ and $(c,d)$ are identical. Next, we find the size of the eigenvalue support of $(a,b),$ whenever $a+b=\frac{n+1}{2}.$
\begin{lem}\label{2l2}
Let $n$ be an odd positive integer with $n>8.$ If $a$ and $b$ are two vertices of $P_n$ satisfying $a+b=\frac{n+1}{2},$ then cardinality of the eigenvalue support of $(a,b)$ is at least $5.$
\end{lem}
\begin{proof}
Without loss of generality let $a>b.$ Since $a+b=\frac{n+1}{2},$ we must have $a,b\in\cb{1,2,\ldots,\frac{n-1}{2}},$ and the maximum value of $a-b$ is $\frac{n-3}{2}.$ If $\lambda_j=2\cos{\ob{\frac{j\pi}{n+1}}}\notin\sigma_{ab}$ for some $j\in\cb{1, 3, 4, 5, 7, 8}$ then \[\cos{\ob{\frac{(a+b)j\pi}{2(n+1)}}}\sin{\ob{\frac{(a-b)j\pi}{2(n+1)}}}=0.\]
Since $\cos{\ob{\frac{(a+b)j\pi}{2(n+1)}}}=\cos{\ob{\frac{j\pi}{4}}}\neq 0,$ we must have $\frac{(a-b)j}{2(n+1)}\in\Zl.$ Consequently, we have $\lambda_1, \lambda_3, \lambda_4 \in \sigma_{ab}.$ If $\frac{(a-b)j}{2(n+1)}$ is an integer for some $j\in\cb{5,7,8},$ then we necessarily have $\frac{(a-b)j}{2(n+1)}=1.$ Hence $\sigma_{ab}$ contains at least two among $\lambda_5, \lambda_7, \lambda_8$, and the result follows.
\end{proof}

The Euler totient function $\phi{(n)}$ counts positive integers that are less than and co-prime to $n.$ One may observe that $\phi{(n)}\geq \sqrt{\frac{n}{2}},$ for all $n.$ The eigenvalue $\lambda_2$ of the path $P_n$ is an algebraic integer of degree $\frac{\phi(n+1)}{2}$ over the rational numbers \cite[Theorem 1]{leh}. Consequently, $\frac{\phi(n+1)}{2}>2$ except for a few initial values of $n.$ We use Theorem \ref{2t_1} to deduce the following result.

\begin{thm}
    A path $P_n$ on $n$ vertices admits perfect pair state transfer if and only if $n\in\cb{3,5,7}.$
\end{thm}

\begin{proof}
    Since periodicity is a necessary condition for the existence of PPST, Lemma \ref{2l1} implies that only $P_4$ may have PPST among all paths on even number of vertices. In \cite[Corollary 3.4]{kim}, we find that if a pair state in $P_n$ is periodic, then the size of its eigenvalue support is at most $4.$ Now, Lemma \ref{2l1} and Lemma \ref{2l2} together implies that there is no periodic pair state in $P_n$ whenever $n$ is odd and $n\geq 13.$ Therefore, it is enough to consider the cases when $n\in\cb{3, 4, 5, 7, 9, 11}.$\\  
{\bf Case  $n\in\cb{3, 5, 7}$:} The path $P_3$ has PST between the end vertices at $\frac{\pi}{\sqrt{2}}$, and is periodic at the internal vertex at the
same time. Therefore, $P_3$ exhibits PPST at  $\frac{\pi}{\sqrt{2}}$ between $(1,2)$ and $(2,3).$ In Section \ref{2s_3}, we find that $P_5$ exhibits PPST at  $\frac{\pi}{2}$ between $(1,5)$ and $(2,4).$ However, there is no PPST from the remaining pairs, as their support contains $\sqrt{3}$ along with one of the eigenvalues $-1$ or $1.$ In the case of $P_7$, we use Theorem \ref{2c2} to conclude that $P_7$ admits PPST at $\frac{\pi}{\sqrt{2}}$ between $(1,7)$ and $(3,5).$ The eigenvalue supports of the remaining pairs except $(2,6)$ contain the eigenvalue $2\cos{\frac{\pi}{8}},$ which results no periodicity in these pairs. Since PPST is monogamous, there is no PPST from $(2,6)$ as well.
 \\
 {\bf Case $n\in \cb{4, 9, 11}$:} One may use Proposition \ref{2c1} to conclude that there is no PPST in $P_4$ between $(1,4)$ and $(2,3).$ Then it is enough to analyze PPST in $P_4$ from $(1,2)$, $(1,3)$, $(2,4)$ and $(3,4).$ Since 
 \[\sigma_{12}=\left\{\frac{-1\pm\sqrt{5}}{2}, \frac{1\pm\sqrt{5}}{2}\right\}=\sigma_{13}=\sigma_{24}=\sigma_{34},\]
there is no PPST in $P_4.$ In the case of $P_9,$ one may observe from Theorem \ref{2c2} that there is no PPST from $(1,9),$ $(2,8),$ $(3,7)$ and $(4,6)$ as there is no PST in $P_4.$ The eigenvalue supports of remaining pairs contain $2\cos{\frac{3\pi}{10}}$ resulting no PPST in $P_9$. Now we consider $n=11,$ and observe that $P_{11}$ admits no PPST from $(1,11),$ $(2,10),$  $(3,9),$ $(4,8),$ $(5,7),$ as there is no PST in $P_5.$ The eigenvalue support of the remaining pairs contain $2\cos{\frac{5\pi}{12}}$ resulting no PPST in $P_{11}.$
 \end{proof}

\begin{figure}
\begin{multicols}{3}
                    \begin{tikzpicture}[scale=0.4,auto=left]
                       \tikzstyle{every node}=[circle, thick, black!90, fill=white, scale=0.65]
               \node[draw,minimum size=0.55cm, 
                inner sep=0 pt] (1) at (1,0) {};
			\node[draw,minimum size=0.55cm, inner sep=0 pt] (2) at (3,0) {$a$};
				\node[draw,minimum size=0.55cm, inner sep=0 pt] (3) at (5,0) {$c$};
    \node[draw,minimum size=0.55cm, 
                inner sep=0 pt] (4) at (7,0) {};
			\node[draw,minimum size=0.55cm, inner sep=0 pt] (5) at (9,0) {$d$};
				\node[draw,minimum size=0.55cm, inner sep=0 pt] (6) at (11,0) {$b$};
			\draw[thick, black!90](1)-- (2)-- (3)--(4)-- (5)-- (6);	
			
			\draw[thick,black!70, dashed={on 5pt off 2.5pt}] (1.2,0.4) .. controls (4,2) and (8,2) .. (10.8,0.4);
                \end{tikzpicture}

                  \begin{tikzpicture}[scale=0.4,auto=left]
                       \tikzstyle{every node}=[circle, thick, black!90, fill=white, scale=0.65]
               \node[draw,minimum size=0.55cm, 
                inner sep=0 pt] (1) at (1,0) {$a$};
			\node[draw,minimum size=0.55cm, inner sep=0 pt] (2) at (3,0) {$c$};
				\node[draw,minimum size=0.55cm, inner sep=0 pt] (3) at (5,0) {};
    \node[draw,minimum size=0.55cm, 
                inner sep=0 pt] (4) at (7,0) {};
			\node[draw,minimum size=0.55cm, inner sep=0 pt] (5) at (9,0) {$d$};
				\node[draw,minimum size=0.55cm, inner sep=0 pt] (6) at (11,0) {$b$};
			\draw[thick, black!90](1)-- (2)-- (3)--(4)-- (5)-- (6);	
			
			\draw[thick,black!70, dashed={on 5pt off 2.5pt}] (5.4,0.2) .. controls (6.5,2) and (7.5,2) .. (8.6,0.2);
   \draw[thick,black!70, dashed={on 5pt off 2.5pt}] (3.4,-0.2) .. controls (4.5,-2) and (5.5,-2) .. (6.6,-0.2);
                \end{tikzpicture}

                   \begin{tikzpicture}[scale=0.4,auto=left]
                       \tikzstyle{every node}=[circle, thick, black!90, fill=white, scale=0.65]
               \node[draw,minimum size=0.55cm, 
                inner sep=0 pt] (1) at (1,0) {};
			\node[draw,minimum size=0.55cm, inner sep=0 pt] (2) at (3,0) {$a$};
				\node[draw,minimum size=0.55cm, inner sep=0 pt] (3) at (5,0) {$c$};
    \node[draw,minimum size=0.55cm, 
                inner sep=0 pt] (4) at (7,0) {};
			\node[draw,minimum size=0.55cm, inner sep=0 pt] (5) at (9,0) {$b$};
				\node[draw,minimum size=0.55cm, inner sep=0 pt] (6) at (11,0) {$d$};
			\draw[thick, black!90](1)-- (2)-- (3)--(4)-- (5)-- (6);	
			
			\draw[thick,black!70, dashed={on 5pt off 2.5pt}] (1.2,0.35) .. controls (3.5,2) and (6.5,2) .. (9.2,0.35);
            
   \draw[thick,black!70, dashed={on 5pt off 2.5pt}] (3.2,-0.35) .. controls (4.5,-1) and (5.5,-1) .. (6.8,-0.35);
   
    \draw[thick,black!70, dashed={on 5pt off 2.5pt}] (7.2,-0.35) .. controls (8.5,-1) and (9.5,-1) .. (10.8,-0.35);
                \end{tikzpicture}

                \end{multicols}
		\caption{PPST in $P_6$ with additional edges.}
  \label{fig2}
	\end{figure}
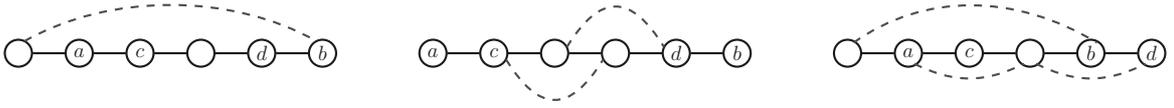

 Although $P_6$ does not exhibit PPST, we use Theorem \ref{2c2} to observe that introducing additional edges to $P_6$ (see Figure \ref{fig2}) gives PPST between $(a,b)$ and $(c,d)$ in the resulting graph. Of course, one may apply Theorem \ref{2c2} to note that there are several ways to have PPST in $P_n$ with a few additional edges. 
 
  \begin{thm}
 Let $P_n$ be a path on $n$ vertices with $n\geq 6,$ where disjoint union of two isomorphic copies of either $P_2$ or $P_3$ appearing as an induced subgraph. Then PPST can be achieved in $P_n$ by adding at most $4$ edges. 
 \end{thm}
 
 A characterization of Laplacian pretty good pair state transfer can be found in \cite{wang2}. Here we provide a sufficient condition for the existence of pair-PGST in $P_n$ relative to the adjacency matrix. In \cite{god4}, Godsil et al. showed that PGST occurs between the end vertices of $P_n$ if and only if $n+1 = p,~2p,$ where $p$ is a prime, or $n+1 = 2^m.$ Moreover, when PGST occurs between the end vertices of $P_n$, then it occurs between vertices $a$ and $n+1-a$ for all $a\neq \frac{n+1}{2}.$ In \cite{bom}, van Bommel complete the characterization of PGST on paths as follows.

\begin{thm}\cite{bom}\label{bom}
 There is pretty good state transfer on $P_n$ between vertices $a$ and $b$ if and only if $a+b=n+1$ and either:
  \begin{enumerate}
  \item $n+1=2^k$, where $k$ is a positive integer, or
  \item $n+1=2^kp,$ where $k$ is a non-negative integer and $p$ is an odd prime, and $a$ is a multiple
 of $2^{k-1}.$
  \end{enumerate}
\end{thm}

If PGST occurs in a path between $a$ and $b,$ and also between $c$ and $d$ with respect to same time sequence, then one has pair-PGST between $(a,c)$ and $(b,d).$ Theorem \ref{bom} in combination with Theorem \ref{2c2} gives more pair of vertices exhibiting pair-PGST, where a pair of $P_n$ can be realized to appear as isomorphic branches of $P_{2n+1}.$ 

 \begin{thm}\label{4c1}
   A path on $n$ vertices, $P_n$ exhibits pretty good pair state transfer between $(a,n+1-a)$ and $(\frac{n+1}{2}-a,\frac{n+1}{2}+a),$ whenever $a<\frac{n+1}{2}$ with $a\neq\frac{n+1}{4}$ and either:
 \begin{enumerate}
     \item $n+1= 2^{k},$ where $k>2$ is a positive integer, or
     \item $n+1=2^k p,$ where $k$ is a positive integer and $p$ is an odd prime, and $a$ is a multiple
 of $2^{k-2}.$
 \end{enumerate}
 \end{thm}
It may be of some interest to find whether there is pair-PGST in $P_n$ except the cases mentioned above. Observe that the Cartesian products $P_3\square P_n$ and $C_4\square P_n$ have two copies of $P_n$ as isomorphic branches. Using Theorem \ref{2c2} and Theorem \ref{bom}, one obtains pair-PGST in the Cartesian products. Moreover, let $G$ be an arbitrary graph. Consider $G+P_n,$ the join of  $G$ with $P_n,$ $G\circ P_n,$ the corona product \cite{bar1} of $G$ with $P_n,$ and other such possible cases where the isomorphic branches of $P_n$ are also isomorphic branches of the graph. Theorem \ref{4c1} applies to provide more class of graphs having pair-PGST, including the join $G+P_n,$ the corona $G\circ P_n,$ and others.

 \subsection{Cycles}
  Let $C_n$ be a cycle with vertex set $\mathbb{Z}_n$, where two vertices $j$ and $k$ are adjacent if and only if $j-k\equiv\pm 1\text{ mod }n.$ Observe that a pair of $P_{n-1}$ appear as isomorphic branches of $C_{2n}.$
  It follows from Theorem \ref{2c2} that both $C_6$ and $C_8$ exhibit PPST. In \cite [Theorem 6.5]{kim}, we find that $C_4, C_6,$ and $C_8$ are the only cycles that exhibit PPST. Although there is no PPST in $C_{10},$ we observe using Theorem \ref{2c2} that introducing additional edges to $C_{10}$ (see figure \ref{2fig3})  gives PPST between pairs $(a,b)$ and $(c,d)$ in the resulting graph. The next result finds cycles with a few additional edges exhibiting PPST.

 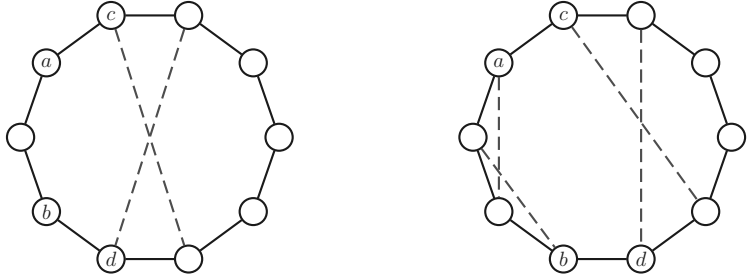
\begin{figure}
		\centering
                    \begin{tikzpicture}[scale=1.7,auto=left]
                       \tikzstyle{every node}=[circle, thick, black!90, fill=white, scale=0.65]

                \node[draw,minimum size=0.55cm, 
                inner sep=0 pt] (1) at (1, 0) {};
				 \node[draw,minimum size=0.55cm, inner sep=0 pt] (2) at (0.80, 0.58) {};
				 \node[draw,minimum size=0.55cm, inner sep=0 pt] (3) at (0.30, 0.95) {};
				 \node[draw,minimum size=0.55cm, 
                 inner sep=0 pt] (4) at (-0.30, 0.95) {$c$};
				 \node[draw,minimum size=0.55cm, inner sep=0 pt] (5) at (-0.80, 0.58) {$a$};
				 \node[draw,minimum size=0.55cm, inner sep=0 pt] (6) at (-1,0) {};
                 \node[draw,minimum size=0.55cm, 
                 inner sep=0 pt] (7) at (-0.80,-0.58) {$b$};
				 \node[draw,minimum size=0.55cm, inner sep=0 pt] (8) at (-0.30,-0.95) {$d$};
				 \node[draw,minimum size=0.55cm, inner sep=0 pt] (9) at (0.30, -0.95) {};
                 \node[draw,minimum size=0.55cm, 
                 inner sep=0 pt] (10) at (0.80, -0.58) {};
                \draw[thick, black!90] (1)-- (2)--(3)--(4)-- (5)-- (6)--(7)--(8)--(9)--(10)--(1);
                  \draw[thick,black!70, dash pattern={on 5pt off 2.5pt}](3)--(8);
      \draw[thick,black!70, dash pattern={on 5pt off 2.5pt}](4)--(9);

      \node[draw,minimum size=0.55cm, 
                inner sep=0 pt] (11) at (4.5, 0) {};
				 \node[draw,minimum size=0.55cm, inner sep=0 pt] (12) at (4.30, 0.58) {};
				 \node[draw,minimum size=0.55cm, inner sep=0 pt] (13) at (3.80, 0.95) {};
				 \node[draw,minimum size=0.55cm, 
                 inner sep=0 pt] (14) at (3.2, 0.95) {$c$};
				 \node[draw,minimum size=0.55cm, inner sep=0 pt] (15) at (2.7, 0.58) {$a$};
				 \node[draw,minimum size=0.55cm, inner sep=0 pt] (16) at (2.5,0) {};
                 \node[draw,minimum size=0.55cm, 
                 inner sep=0 pt] (17) at (2.7,-0.58) {};
				 \node[draw,minimum size=0.55cm, inner sep=0 pt] (18) at (3.2,-0.95) {$b$};
				 \node[draw,minimum size=0.55cm, inner sep=0 pt] (19) at (3.80, -0.95) {$d$};
                 \node[draw,minimum size=0.55cm, 
                 inner sep=0 pt] (20) at (4.30, -0.58) {};
                \draw[thick, black!90] (11)-- (12)--(13)--(14)-- (15)-- (16)--(17)--(18)--(19)--(20)--(11);
                  \draw[thick,black!70, dash pattern={on 5pt off 2.5pt}](13)--(19);
      \draw[thick,black!70, dash pattern={on 5pt off 2.5pt}](14)--(20);
       \draw[thick,black!70, dash pattern={on 5pt off 2.5pt}](15)--(17);       
 \draw[thick,black!70, dash pattern={on 5pt off 2.5pt}](16)--(18);

\end{tikzpicture}	
		\caption{PPST in $C_{10}$ with additional edges.}
  \label{2fig3}
	\end{figure}
	
 \begin{thm}
 Let $C_n$ be a cycle on $n$ vertices with $n>6,$ where disjoint union of two isomorphic copies of $P_2$ (or $P_3$) appearing as an induced subgraph. Then PPST can be achieved in $C_n$ by adding at most $4$ edges. 
 \end{thm}

Before we begin the investigation of pair-PGST in cycles, we have the following observation. The map which sends $j$ to $-j$ is an automorphism of $C_{2n},$ fixing only the vertices $0$ and $n.$ The following result is now immediate from Lemma \ref{al1}.
\begin{lem}\label{2l3}
There is no pair-PGST from a pair of antipodal vertices in an even cycle.
\end{lem}
We find in \cite{pal4} that a cycle $C_n$ exhibits PGST if and only if $n=2^k$ for $k\geq 2,$ and it occurs between the pair of antipodal vertices $a$ and $a+\frac{n}{2}.$ As a natural consequence, we have the following.
\begin{thm}\label{2c4}
Let $n=2^k$ where $k\geq 2,$ and consider a cycle $C_n$ having two vertices $a$ and $b.$ Pretty good pair state transfer occurs in $C_n$ between $(a,b)$ and $(a+\frac{n}{2}, b+\frac{n}{2})$ if and only if $a-b\neq \frac{n}{2}.$ 
\end{thm}
  Suppose $\omega_n = \exp{(\frac{2\pi i}{n})}$ is the primitive $n$-th root of unity. The eigenvalues of $C_n$ are given by $\lambda_j = 2\cos \frac{2j\pi}{n},$ and the corresponding eigenvectors are
 $\left[1, \omega^j_n, {\omega^{2j}_n},\ldots, {\omega^{(n-1)j}_n}\right]^T,$ where $j\in\Zl_n.$ The characterization of PGST on paths in Theorem \ref{bom} together with Theorem \ref{2c2} leads to the following conclusion.
 \begin{thm}\label{2c9}
    The cycle $C_n$ exhibits pretty good pair state transfer between $(a,n-a)$ and $(\frac{n}{2}-a,\frac{n}{2}+a),$ whenever $0<a<\frac{n}{2}$ with $a\neq\frac{n}{4}$ and either:
    \begin{enumerate}
    \item $n = 2^{k},$ where $k$ is a positive integer greater than two or,
    \item $n=2^k p,$ where $k$ is a positive integer and $p$ is an odd prime, and $a=2^{k-2}r,$ for some positive integer $r.$
    \end{enumerate}   
\end{thm}
Using the automorphisms of $C_n,$ one can observe in Theorem \ref{2c9} that if pair-PGST occurs from $(a,n-a)$ then it must occur from each pair of vertices at a distance $2a,$ where $0<a<\frac{n}{4}.$ Before we investigate pair-PGST in the remaining cycles, we include the following observation. Suppose the premise of Lemma \ref{2t1} holds, and pair-PGST occurs  between $\ob{a,\hat{f}(a)}$ and $(b,c)$ in $G.$ Then there exists $t_k\in\mathbb{R}$ and $\gamma\in\Cl$ such that 
\[\displaystyle\lim_{k\to\infty} U_G\ob{t_k}\ob{\e_a-\e_{\hat{f}(a)}}=\gamma\ob{\e_b-\e_c}.\]
Using the permutation matrix $P$ corresponding to the automorphism $\hat{f},$ we have
\[\gamma \ob{\e_{\hat{f}(b)}-\e_{\hat{f}(c)}}=\displaystyle\lim_{k\to\infty} U_G\ob{t_k}P\ob{\e_a-\e_{\hat{f}(a)}}=-\gamma\ob{\e_b-\e_c}.\]
Hence $c=\hat{f}(b)\neq b,$ and the orbits of $\hat{f}$ containing $b$ and $c$ are identical. Then either $b$ or $c$ is a vertex of $X_1,$ and the following result holds.  
\begin{lem}\label{2l4}
Suppose the premise of Theorem \ref{2c2} holds. If pretty good pair state transfer occurs from $\ob{a,\hat{f}(a)},$ for some $a\in X_1,$ then there exists $b\in X_1$ such that it occurs between $\ob{a,\hat{f}(a)}$ and $\ob{b,\hat{f}(b)}.$ Moreover, the graph $X_1$ admits pretty good state transfer between vertices $a$ and $b.$     
\end{lem}

Now we provide a complete characterization of pair-PGST on cycles. We show that $C_4$ along with the cycles stated in Theorem \ref{2c9} are the only possible cycles exhibiting pair-PGST.

\begin{thm}\label{ct1}
    A cycle $C_n$ on $n$ vertices admits pretty good pair state transfer if and only if either $n=2^k$ or $n=2^k p,$ where $k$ is a positive integer and $p$ is an odd prime.
\end{thm}
\begin{proof} Since Lemma \ref{al2} holds, there is no pair-PGST in $C_n$ whenever $n$ is odd. It is enough to consider the cases where $n$ is even and $n\notin\cb{ 2^k,2^k p},$ where $k$ is a positive integer and $p$ is an odd prime.\\ 
\textbf{Case 1.} Let $a$ and $c$ be two vertices in $C_n$ at a distance even. Two paths $P$ and $Q$ each having $\frac{n}{2}-1$ vertices appear as isomorphic branches of $C_n,$ where $a\in V(P),$ $c\in V(Q)$ and the switching automorphism $\hat{f}$ of $C_n$ satisfying $c=\hat{f}(a).$ If there is pair-PGST from $(a,c)$, then by Lemma \ref{2l4}, there exists $b\in V(P)$ such that PGST occurs between $a$ and $b.$ Since the path $P$ has $\frac{n}{2}-1$ vertices, there is no PGST in $P$ as in Theorem \ref{bom}. As a consequence, $C_n$ has no pair-PGST in this case.\\
\textbf{Case 2.} Let $C_n$ have pair-PGST from $(a,b),$ where $a$ and $b$ are at a distance odd. There is a rotation $f$ of $C_n$ satisfying $f(a)=b.$ By Lemma \ref{al1}, there exists pair-PGST from $\ob{a,f(b)}$ as well. The distance between $a$ and $f(b)$ is now even, a contradiction. \end{proof}

 The Cartesian products $P_3\square C_n$ and $C_4\square C_n$ have two copies of $C_n$ as isomorphic branches. Using Theorem \ref{2c2} and \cite[Theorem 13]{pal4}, one obtains pair-PGST in the Cartesian products. Moreover, let $G$ be an arbitrary graph. Consider $G+C_n,$ the join of $G$ with $C_n,$ $G\circ C_n,$ the corona product of $G$ with $C_n,$ and other such possible cases where the isomorphic branches of $C_n$ are also isomorphic branches of the graph. Theorem \ref{4c1} applies to provide more classes of graphs having pair-PGST, including the join $G+C_n,$ the corona $G\circ C_n,$ and others.

\section{Multi-state Transfer}\label{2s5}

Let $G$ be a graph on $n$ vertices having the characteristic vectors $\e_1,\e_2,\ldots,\e_n\in\Cl^n$. An $(m,L)$-state of $G$ is a linear combination $l_1\e_{j_1}+l_2\e_{j_2}+\cdots+l_m\e_{j_m},$
where $L=\ob{l_1, l_2, \ldots, l_m}\in\Cl^m$ with $\displaystyle\sum_{k=1}^m |l_k|^2=1.$ One may consider perfect $(m,L)$-state transfer in graphs where vertex states are replaced by $(m,L)$-states in the definition of PST. Let $X_1,X_2,\ldots, X_m$ be pairwise isomorphic branches of a graph $G$ with isomorphisms $f_k:X_1\to X_k,$ where $k=2, 3,\ldots,m.$ The pair states $\frac{1}{\sqrt{2}}\ob{\e_a-\e_{f_k(a)}}$ span the $(m,L)$-state $l_1\e_a+l_2\e_{f_2(a)}+\cdots+l_m\e_{f_m(a)}$ for every choice of $l_1, l_2, \ldots, l_m\in\Cl$ with $\displaystyle\sum_{k=1}^m l_k=0.$ 

   \begin{figure}
		\centering
                    \begin{tikzpicture}[scale=1.5,auto=left]
                       \tikzstyle{every node}=[circle, thick, black!90, fill=white, scale=0.65]

                 \node[draw,minimum size=0.55cm, inner sep=0 pt] (1) at (0, 0) {$u$};
				 \node[draw,minimum size=0.55cm, inner sep=0 pt] (2) at (0, 1.5) {$v$};

                 \node[draw,minimum size=0.55cm, inner sep=0 pt] (3) at (-1.15, 0.5) {$a_1$};
				 \node[draw,minimum size=0.55cm, 
                 inner sep=0 pt] (4) at (-1.15, 2) {$b_1$};

                 \node[draw,minimum size=0.55cm, inner sep=0 pt] (5) at (-0.8, 0.75) {$a_2$};
				 \node[draw,minimum size=0.55cm, inner sep=0 pt] (6) at (-0.8,2.25) {$b_2$};

                 \node[draw,minimum size=0.55cm, inner sep=0 pt] (7) at (-0.45, 1) {$a_3$};
				 \node[draw,minimum size=0.55cm, inner sep=0 pt] (8) at (-0.45,2.5) {$b_3$};

                 \node[draw,minimum size=0.55cm, 
                 inner sep=0 pt] (9) at (1.15,0.5) {$a_l$};
				 \node[draw,minimum size=0.55cm, inner sep=0 pt] (10) at (1.15,2) {$b_l$};
				
                \draw[thick, black!90] (1)-- (2)--(6)--(5)-- (1);
 \draw[thick, black!90]  (2)--(4)-- (3)-- (1)--(9)--(10)--(2);
                \draw[thick, black!90] (1)-- (7)--(8)--(2);

                \draw[thick,black!70, dotted={on 5pt off 2.5pt}] (7).. controls (-0.2,1) and (0.5,1.2)  .. (9);
                \draw[thick,black!70, dotted={on 5pt off 2.5pt}] (8).. controls (-0.2,2.5) and (0.5,2.7)  .. (10);

\end{tikzpicture}	
		\caption{Perfect multi-state transfer in book graph.}
  \label{Fig2}
	\end{figure}
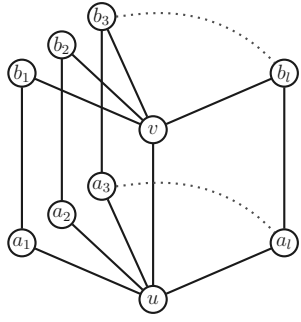

 Consider the book graph $K_{1,l}\square P_2,$ as shown in Figure \ref{Fig2}, where $K_{1,l}$ is a star on $l+1$ vertices. Here the paths with end vertices $a_k$ and $b_k$ are pairwise isomorphic branches of $K_{1,l}\square P_2.$ Let $U_l(t)$ be the transition matrix of $K_{1,l}\square P_2.$ Assuming $l\geq 3$ and using Theorem \ref{2c2}, we have $U_l\ob{\frac{\pi}{2}}(\e_{a_1}-\e_{a_3})=i(\e_{b_1}-\e_{b_3})$ and $U_l\ob{\frac{\pi}{2}}(\e_{a_2}-\e_{a_3})=i(\e_{b_2}-\e_{b_3}).$ Hence, the book graph exhibits perfect $(m,L)$-state transfer with $m=3$ and $L=(1,1,-2),$ since
 $U_l\ob{\frac{\pi}{2}}(\e_{a_1}+\e_{a_2}-2\e_{a_3})=i(\e_{b_1}+\e_{b_2}-2\e_{b_3}).$ The next result follows from Theorem \ref{2c2}.  
\begin{thm}\label{2t3} 
 Let $X_1,X_2,\ldots, X_m$ be pairwise isomorphic branches of a graph $G$ with isomorphisms $f_k:X_1\to X_k,$ where $k=2, 3,\ldots,m.$ If perfect state transfer occurs in $X_1$ between $a$ and $b,$ then $G$ exhibits perfect $(m,L)$-state transfer at the same time between $l_1\e_a+\displaystyle\sum_{k=2}^m l_k\e_{f_k(a)}$ and $l_1\e_b+\displaystyle\sum_{k=2}^m l_k\e_{f_k(b)}$ whenever $\displaystyle\sum_{k=1}^m l_k=0.$
\end{thm}

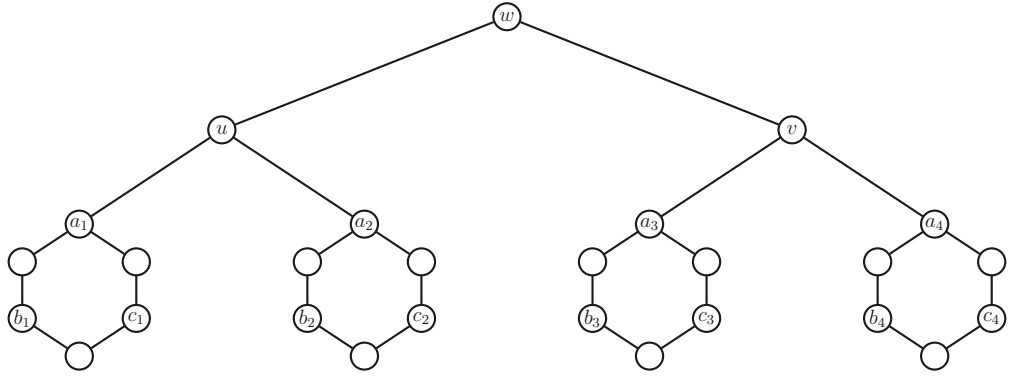
\begin{figure}
		\centering
                    \begin{tikzpicture}[scale=.5,auto=left]
                       \tikzstyle{every node}=[circle, thick, black!90, fill=white, scale=0.65]
                    
                \node[draw,minimum size=0.55cm, 
                inner sep=0 pt] (1) at (0,0) {$c_1$};
				  \node[draw,minimum size=0.55cm, inner sep=0 pt] (2) at (0,1.5) {};
				  \node[draw,minimum size=0.55cm, 
                  inner sep=0 pt] (3) at (-1.5,2.5)  {$a_1$};
                  \node[draw,minimum size=0.55cm, 
                  inner sep=0 pt] (4) at (-3,1.5) {};
                  \node[draw,minimum size=0.55cm, inner sep=0 pt] (5) at (-3, 0) {$b_1$};
                  \node[draw,minimum size=0.55cm, inner sep=0 pt] (6) at (-1.5, -1) {};

                  \draw[thick, black!90] (1)-- (2)--(3)--(4)--(5)--(6)--(1);

                   \node[draw,minimum size=0.55cm, 
                inner sep=0 pt] (7) at (7.5,0) {$c_2$};
				  \node[draw,minimum size=0.55cm, inner sep=0 pt] (8) at (7.5,1.5) {};
				  \node[draw,minimum size=0.55cm, 
                  inner sep=0 pt] (9) at (6,2.5)  {$a_2$};
                  \node[draw,minimum size=0.55cm, 
                  inner sep=0 pt] (10) at (4.5,1.5) {};
                  \node[draw,minimum size=0.55cm, inner sep=0 pt] (11) at (4.5, 0) {$b_2$};
                  \node[draw,minimum size=0.55cm, inner sep=0 pt] (12) at (6, -1) {};
                   \draw[thick, black!90] (7)-- (8)--(9)--(10)--(11)--(12)--(7);

                     \node[draw,minimum size=0.55cm, 
                inner sep=0 pt] (13) at (15,0) {$c_3$};
				  \node[draw,minimum size=0.55cm, inner sep=0 pt] (14) at (15,1.5) {};
				  \node[draw,minimum size=0.55cm, 
                  inner sep=0 pt] (15) at (13.5,2.5)  {$a_3$};
                  \node[draw,minimum size=0.55cm, 
                  inner sep=0 pt] (16) at (12,1.5) {};
                  \node[draw,minimum size=0.55cm, inner sep=0 pt] (17) at (12, 0) {$b_3$};
                  \node[draw,minimum size=0.55cm, inner sep=0 pt] (18) at (13.5, -1) {};

                  \draw[thick, black!90] (13)--(14)--(15)--(16)--(17)--(18)--(13);

                   \node[draw,minimum size=0.55cm, 
                inner sep=0 pt] (19) at (22.5,0) {$c_4$};
				  \node[draw,minimum size=0.55cm, inner sep=0 pt] (20) at (22.5,1.5) {};
				  \node[draw,minimum size=0.55cm, 
                  inner sep=0 pt] (21) at (21,2.5)  {$a_4$};
                  \node[draw,minimum size=0.55cm, 
                  inner sep=0 pt] (22) at (19.5,1.5) {};
                  \node[draw,minimum size=0.55cm, inner sep=0 pt] (23) at (19.5, 0) {$b_4$};
                  \node[draw,minimum size=0.55cm, inner sep=0 pt] (24) at (21, -1) {};
                   \draw[thick, black!90] (19)-- (20)--(21)--(22)--(23)--(24)--(19);

9
                  \node[draw,minimum size=0.55cm, 
                  inner sep=0 pt] (25) at (2.25,5)  {$u$};
                  \node[draw,minimum size=0.55cm, 
                  inner sep=0 pt] (26) at (17.25,5)  {$v$};
                  \node[draw,minimum size=0.55cm, 
                  inner sep=0 pt] (27) at (9.75,8)  {$w$};
                 \draw[thick, black!90] (3)--(25)-- (9);
                 \draw[thick, black!90] (25)-- (27)--(26);
                 \draw[thick, black!90] (15)--(26)-- (21);

\end{tikzpicture}	
		\caption{A graph with perfect multi-state transfer.}
  \label{2Fig7}
	\end{figure}
  Let $G$ be the graph as given in Figure \ref{2Fig7}. Two copies of $C_6$ are appearing as isomorphic branches in $G\setminus\cb{w}.$ The cycle $C_6$ has perfect $s$-pair state transfer \cite[Theorem 6.5]{kim} at time $\pi$ between $\e_{a_1}+\frac{1}{2}\e_{b_1}$ and $\e_{c_1}+\frac{1}{2}\e_{b_1}.$ By Theorem \ref{2c2}, the graph $G\setminus\cb{w}$ exhibits perfect state transfer in  between $\e_{a_1}-\e_{a_2}+\frac{1}{2}\ob{\e_{b_1}-\e_{b_2}}$ and $\e_{c_1}-\e_{c_2}+\frac{1}{2}\ob{\e_{b_1}-\e_{b_2}}.$ Again we apply Theorem \ref{2c2} to obtain perfect state transfer in $G$ between
  $\e_{a_1}-\e_{a_2}-\e_{a_3}+\e_{a_4}+$ $\frac{1}{2}\ob{\e_{b_1}-\e_{b_2}-\e_{b_3}+\e_{b_4}}$ and $\e_{c_1}-\e_{c_2}-\e_{c_3}+\e_{c_4}+\frac{1}{2}\ob{\e_{b_1}-\e_{b_2}-\e_{b_3}+\e_{b_4}}.$ The following result provides a framework for constructing infinite family of graphs exhibiting multi-state transfer. The proof is immediate from Theorem \ref{2c2}.
\begin{thm}\label{2t9}
    Suppose the premise of Theorem \ref{2c2} holds. If $X_1$ has perfect state transfer between $\mu$ and $\zeta$, then $G$ exhibits perfect state transfer between $(\mu-f(\mu))$ and $(\zeta-f(\zeta)).$  
\end{thm}

 \section*{Disclosure statement}
  No potential conflict of interest was reported by the author(s).
 
\section*{Acknowledgements}
H. Pal is funded by the Science and Engineering Research Board (Project: SRG/2021/000522).
S. Mohapatra is supported by the Department of Science and Technology (INSPIRE: IF210209).


\bibliographystyle{abbrv}
\bibliography{references}

\end{document}